\documentclass[letterpaper,10 pt, conference]{ieeeconf}

\usepackage[letterpaper,margin=0.75in]{geometry}
\IEEEoverridecommandlockouts                   
\usepackage{graphics} % for pdf, bitmapped graphics files

\usepackage{amsmath,amssymb,amsfonts,amsthm} % assumes amsmath package installed
\usepackage{mathtools,bbm}
\usepackage[font={small}]{caption}
\usepackage{subcaption}
\usepackage{bbm}
\usepackage{xcolor}
\usepackage{bm}
\usepackage{bbold}
\usepackage{epsfig, amsbsy}
\usepackage{mathtools}
\usepackage{apptools}
\usepackage{graphicx}
\usepackage{caption}
\usepackage{subcaption}
\usepackage{epstopdf}

\usepackage{changepage}

\usepackage[colorlinks = true,
            linkcolor = blue,
            urlcolor  = blue,
            citecolor = blue,
            anchorcolor = blue]{hyperref}
\usepackage{cleveref}
\theoremstyle{plain}% default
\newtheorem{theorem}{Theorem}[section]
\newtheorem{corollary}{Corollary}[section]

\newtheorem{lemma}{Lemma}[section]

\AtAppendix{\counterwithin{theorem}{subsection}}
\AtAppendix{\counterwithin{lemma}{subsection}}

\def\bs{\boldsymbol}
\def\mbb{\mathbb}
\def\mcal{\mathcal}

\newcommand{\eq}[1]{{#1}^{\rm eq}}

\title{\LARGE \bf
Strategic investments in multi-stage General Lotto games
}
\author{Rahul Chandan$^{\star}$, Keith Paarporn$^{\star}$,  Mahnoosh Alizadeh,  Jason R. Marden % <-this % stops a space
\thanks{ R. Chandan, K. Paarporn, M. Alizadeh, and J. R. Marden are with the Department of Electrical and Computer Engineering at the University of California, Santa Barbara, CA. Contact: \texttt{ \{rchandan,kpaarporn,alizadeh,jrmarden\}@ucsb.edu} This work is supported by UCOP Grant LFR-18-548175, ONR grant \#N00014-20-1-2359, AFOSR grants \#FA9550-20-1-0054 and \#FA9550-21-1-0203, and the Army Research Lab through the ARL DCIST CRA \#W911NF-17-2-0181. $^\star$These authors contributed equally to this work.}
}

\begin{document}

\maketitle
\thispagestyle{empty}
\pagestyle{empty}

\begin{abstract}
    In adversarial interactions, one is often required to make strategic decisions over multiple periods of time, wherein decisions made earlier impact a player's competitive standing as well as how choices are made in later stages. In this paper, we study such scenarios in the context of General Lotto games, which models the competitive allocation of resources over multiple battlefields between two players. We propose a two-stage formulation where one of the players has reserved resources that can be strategically  \emph{pre-allocated} across the battlefields in the first stage. The pre-allocation then becomes binding and is revealed to the other player.  In the second stage, the players engage by simultaneously allocating their \emph{real-time} resources against each other. The main contribution in this paper provides complete characterizations of equilibrium payoffs in the two-stage game, revealing the interplay between performance and the amount of resources expended in each stage of the game. We find that real-time resources are at least twice as effective as pre-allocated resources. We then determine the player's optimal investment when there are linear costs associated with purchasing each type of resource before play begins, and there is a limited monetary budget.
    
    % The central problem of interest in this paper is to determine optimal investment decisions, given heterogeneity of the resources available as well as a limited monetary budget. We consider this problem in the context of a two-stage General Lotto game where only one player has the opportunity to invest in two distinct resource types -- \emph{pre-allocated} and \emph{real-time} resources. pre-allocated resources are pre-allocated to battlefields in the first stage, thus giving a player an incumbency advantage when competition subsequently takes place in stage 2. However, pre-allocated resources must remain fixed: there is no opportunity to re-allocate them in stage 2. In contrast, the allocation of real-time resources in stage 2 can be randomized across battlefields. Here, we assume that the opponent only has access to real-time resources, and payoffs are evaluated at the end of stage 2. We completely characterize the optimal mix of pre-allocated and real-time resources, given a fixed monetary budget to invest in both types. 
\end{abstract}

% \begin{keywords}
% game theory, General Lotto games, resource allocation, pre-commitments
% \end{keywords}

\iffalse

\vfill \clearpage 

{\bf Intro:}\\
-- Allocation of resources for cybersecurity, military, etc. (Controls flavour)\\
-- Often exhibits investment into different types of resources, with different costs/capabilities/etc.\\
-- We consider a scenario with two types of resources (real-time and pre-allocated)... Possible interpretations: cheap vs. expensive, early vs. last minute, pre-allocated vs. real-time, public vs. covert\\
-- Understanding how this investment should be done is vital \\
-- Colonel Blotto/General Lotto/Favoritism\\
-- Contributions: (i) we solve for player $A$'s optimal allocation in the GL-F game for any allocation $S$, $M_A$, and obtain a closed-form expression for the payoff; (ii) we identify the optimal division into pre-allocated and real-time resources subject to different types of constraints; (iii)\\

{\bf Model:}\\
-- $A$ allocates $S$ pre-allocated resources, $M_A$ real-time resources\\
-- Two-stage game: $S$, $M_A$ selected in first stage, GL-F played in second stage\\
-- The choice of $S$ and $M_A$ is subject to convex constraint (e.g., linear, norm, cone)\\
-- Example 1: linear cost constraint\\
-- Example 2: max-norm constraint\\
-- Example 3: uncertainty\\

{\bf Main Results:}\\
-- Optimal allocation of $S$ and $M_A$ in Stage 2

\vfill \clearpage 

\fi

\section{Introduction}\label{sec:intro}
% !TEX root = main.tex

% There is significant interest in studying how one should pre-allocate resources to guarantee the best system outcomes against a strategic adversary
% A common feature of these interactions is that decisions are often made before the adversary makes their move,
% Clearly, there resources that can be allocated in real-time have higher value, but these can be expensive and, thus, it is important to understand the trade-off between slow and cheap resources vs. fast and expensive resources

In resource allocation problems, system planners must make investment decisions to mitigate the risks posed by disturbances or strategic interference.  In many practical settings, these investments are made over several stages leading up to the actual time of allocation.  For example, security measures in cyber-physical systems are deployed over long periods of time. As such, attackers can use knowledge of pre-deployed elements to identify vulnerabilities and exploits in the defender's strategy \cite{brown2006defending,zhang2013protecting}.  As another example, power grid operators must bid on forward-capacity (i.e., day-ahead, hour-ahead and real-time) markets to fulfill future demand.  Although grid operators can significantly reduce energy prices and carbon emissions by procuring capacity in day- and hour-ahead markets, they still rely on real-time markets to account for uncertainty in the demand signal \cite{ben2004adjustable,li2020optimal}. Further examples include R\&D contests, team management in competitive sports, and political lobbying \cite{yildirim2005contests}.

% In competitive resource allocation, system planners must make investment decisions to mitigate the risks posed by disturbances or strategic interference.  In many practical settings, these investments are made over several stages leading up to the actual competitive interaction.  For example, security measures in cyber-physical systems are deployed over long periods of time. As such, attackers can use knowledge of pre-deployed elements to identify vulnerabilities and exploits in the defender's strategy \cite{brown2006defending,zhang2013protecting}.  As another example, power grid operators must bid on forward-capacity (i.e., day-ahead, hour-ahead and real-time) markets to fulfill future demand.  Although grid operators can significantly reduce energy prices and carbon emissions by procuring capacity in day- and hour-ahead markets, they still rely on real-time markets to account for uncertainty in the demand signal \cite{ben2004adjustable,li2020optimal}. Further examples include R\&D contests, team management in competitive sports, and political lobbying \cite{yildirim2005contests}.

Indeed, there are numerous real-world examples of systems in which both early and late investments contribute to the system performance.  Notably, many of these scenarios consist of strategic interactions between competitors, and exhibit trade-offs when investing in pre-allocated and real-time resources (e.g., resource costs vs. flexibility in deployment, long-term vs. short-term gains).  In such scenarios, system planners must choose their dynamic investments while accounting for their competitors' decision making, and balancing the trade-offs in early and late investment. % The specific trade-off between early and late investment, and the corresponding optimal investment strategy will strongly depend on the particular problem instance.

% Indeed, there are numerous real-world examples of systems in which early investments contribute to system performance in the eventual competitive interaction.  Many of these scenarios exhibit trade-offs when investing in pre-allocated and real-time resources (e.g., resource costs vs. flexibility in deployment, long-term vs. short-term gains).  Thus, system planners must balance these trade-offs and choose their dynamic investments strategically. The specific trade-off between early and late investment, and the corresponding optimal investment strategy will strongly depend on the particular problem instance.

In this manuscript, we seek to characterize the interplay between early and late investment in competitive resource allocation settings. We pursue this research agenda in the context of General Lotto games, a game-theoretic framework that explicitly describes the competitive allocation of resources between opponents. The General Lotto game is a popular variant of the classic Colonel Blotto game, wherein two budget-constrained players, $A$ and $B$, compete over a set of valuable battlefields. The player that deploys more resources to a battlefield wins its associated value, and the objective for each player is to win as much value as possible. Outcomes in the standard formulations are determined by a single, simultaneous allocation of resources. In the novel formulation introduced in this paper, one of the players can strategically decide how to deploy resources before the actual engagement takes place. The placement of the pre-allocated resources thus has an effect on how the allocation decisions are made at the time of competition.

% Pure strategy equilibria do not exist in these games for most scenarios of interest. Thus, randomizing one's deployment of resources is a central component in the analysis of General Lotto games.

Specifically, we consider the following two-stage scenario. Player $A$ is endowed with $P \geq 0$ resources to be pre-allocated, and both players possess real-time resources $R_A, R_B \geq 0$ to be allocated at the time of competition. In the first stage, player $A$ decides how to \emph{deterministically} deploy the pre-allocated resources $P$ over the battlefields. Player $A$'s endowments and pre-allocation decision then become known to player $B$. In the second stage, both players engage in a General Lotto game where they simultaneously decide how to deploy their real-time resources, and payoffs are subsequently derived. We assume player $B$ does not have any pre-allocated resources at its disposal, and only has real-time resources to compete with. Each player can \emph{randomize} the deployment of her real-time resources. Here, player $B$ must overcome both the pre-allocated and real-time resources deployed by player $A$ to secure a battlefield.  A full summary of our contributions is provided below:
% 

% In \cite{Vu_EC2021}, the authors established the existence of equilibrium strategies and provided computational methods to calculate them to arbitrary precision. Our results in this paper 

% While players only have real-time resources in the standard General Lotto game, we consider a variation in which player $A$ can invest in real-time resources ($M_A$) \emph{and} pre-allocated resources ($S$), while player $B$ only has real-time resources ($M_B$).  Furthermore, this game is played over two stages:  In Stage 1, player $A$ selects her investment $(S,M_A)$, and deploys her pre-allocated resources $S$ over the battlefields; then, in Stage 2, both players deploy their real-time resources $M_A$ and $M_B$.  Player $A$'s decision in Stage 1 is binding and common knowledge in Stage 2.  Thus, pre-allocated resources are inherently less effective than real-time resources under this model, as player $B$ knows the deployment of $S$ over the battlefields but only knows the value of $M_A$ when making her own deployment decision.  

% In this context, our contributions are centered around the choice that player $A$ faces in Stage 1:

\smallskip \noindent \textbf{Our Contributions:} Our main contribution in this paper is a full characterization of equilibrium payoffs to both players in our two-stage General Lotto game, given player $A$ has $P$ pre-allocated resources, $R_A$ real-time resources, and player $B$ has $R_B$ real-time resources (Theorem \ref{thm:equilibrium_characterization}). This result also specifies how player $A$ should optimally deploy its pre-allocated resources to the battlefields, each of which has an arbitrary associated value. Our characterization explicitly reveals the relative effectiveness of pre-allocated and real-time resources -- for any desired performance level $\Pi \geq 0$ against $R_B$, we provide the set of all pairs $(P,R_A)$ that achieve the payoff $\Pi$ for player $A$ (Theorem \ref{thm:level_set}). As a consequence, we show that, to achieve the same performance $\Pi > 0$ using only one type of resource, player $A$ needs at least double the amount of pre-allocated resources than the amount of real-time resources (Corollary \ref{cor:twice_as_effective}).

Leveraging the main results, we then derive the optimal investment pair $(P,R_A)$ for player $A$ when there are linear per-unit costs to invest in both types of resources and a limited monetary budget is available. We note that it is optimal to invest in both resources only if the per-unit cost of pre-allocated resources is lower than real-time resources. Indeed, pre-allocated resources are less effective than real-time resources, since their deployment is not randomized and player $B$ has knowledge of their placement. 

% We also characterize the relative effectiveness of pre-allocated and real-time resources in isolation (Corollary \ref{cor:twice_as_effective}). Specifically, we find that to achieve a baseline performance $\Pi > 0$ with only pre-allocated resources, one needs at least twice as many as the amount of real-time resources that achieves the same baseline.

\smallskip \noindent \textbf{Related works:} This manuscript takes preliminary steps towards developing analytical insights about competitive resource allocation in multi-stage scenarios. There is widespread interest in this research objective, where the focus ranges from zero-sum games \cite{Nayyar_2013,Kartik_2021asymmetric,Li_2020}, and dynamic games \cite{Isaacs_1965,VonMoll_2020}, to Colonel Blotto games \cite{Vu_2019combinatorial,Aidt_2019,Shishika_2021,Leon_2021}. The goal of many of these works is to develop computational tools to compute decision-making policies for agents in adversarial and uncertain environments. In comparison, our work provides explicit, analytical characterizations of equilibrium strategies, allowing for insights relating the players' performance with various elements of adversarial interaction to be drawn. As such, our work is related to a recent research thread studying Colonel Blotto games in which allocation decisions are made over multiple stages \cite{Kovenock_2012,Gupta_2014a,Gupta_2014b,Paarporn2021strategically,Chandan_2021_concession}. 

Our work also draws significantly from the primary literature on Colonel Blotto and General Lotto games \cite{Gross_1950,Roberson_2006,Kovenock_2020,Vu_EC2021}. In particular, the simultaneous-move subgame played in the second stage of our model was first proposed by Vu and Loiseau \cite{Vu_EC2021}, and is known as the \emph{General Lotto game with favoritism} (GL-F). Favoritism refers to the fact that pre-allocated resources provide an inherent advantage to one player's competitive chances. Their work establishes existence of equilibria and develops computational methods to calculate them to arbitrary precision. However, this prior work considers pre-allocated resources as exogenous parameters of the game. In contrast, we model the deployment of pre-allocated resources as a strategic element of the competitive interaction.  Furthermore, we provide the first analytical characterization of equilibria and the corresponding payoffs in GL-F games. % Furthermore, the choice of how to deploy pre-allocated resources is a strategic element in our model.

% Computational vs. analytical

% To date, the analysis has focused on one-shot settings
% In this manuscript, we take a preliminary step toward dynamic setting
% Our work draws significantly from the existing literature on Colonel Blotto and General Lotto games.  In particular, we build upon the results put forward in Vu et al. on the General Lotto with favoritism (GL-F) game model.

% Our work also relates to the investment of heterogeneous resource types.  In particular, in settings in which there are both slow, cheap resources and fast, expensive resources available, we provide the mixture of these that guarantees the best system outcome when deployed.
% Vu answers: "How should I adjust my allocation as a function of pre-allocated assets?" We answer: "How should I pre-allocate assets to maximize the effectiveness of my subsequent allocation?"

\section{Problem formulation}\label{sec:model}
% !TEX root = main.tex

\begin{figure*}[t]
    \centering
    \includegraphics[width=\textwidth,trim={0 6.75cm 0cm 0}]{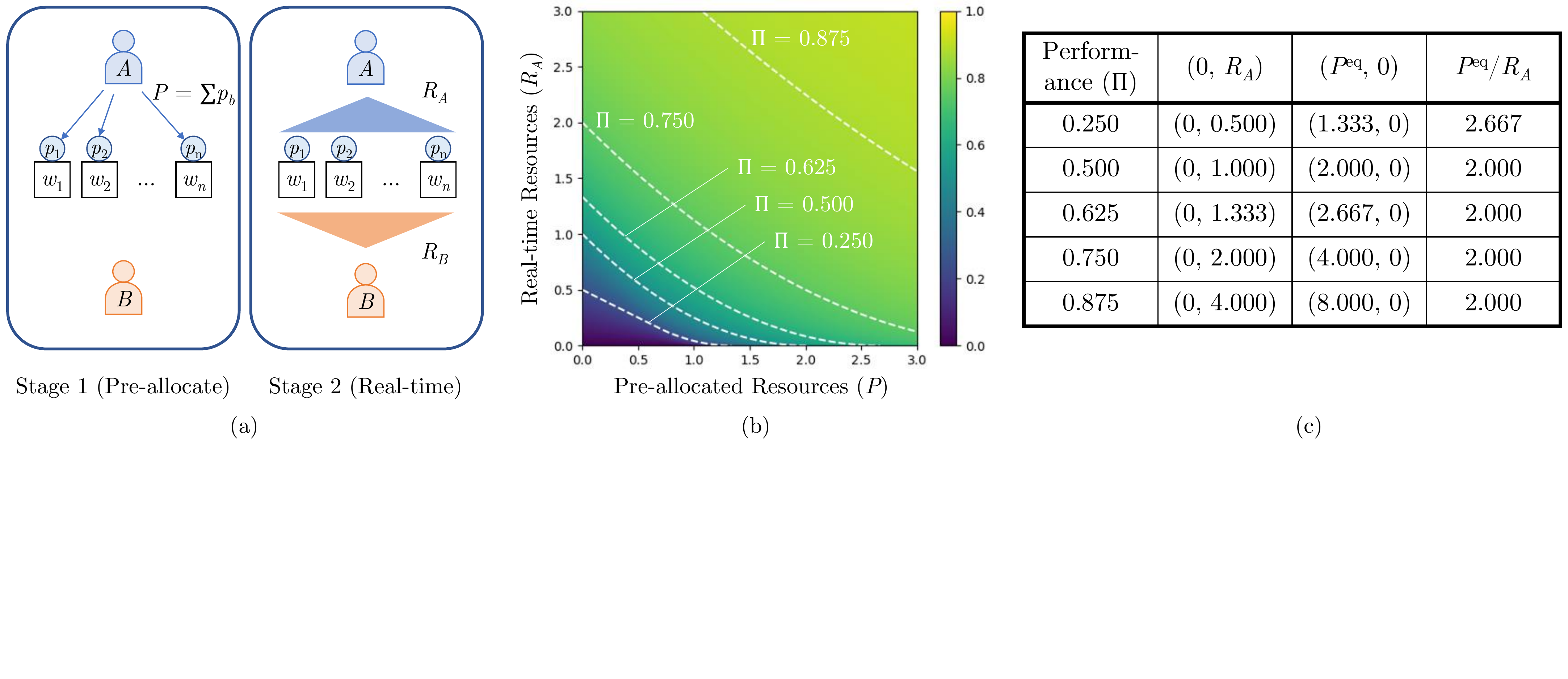}
    \caption{(a) The two-stage General Lotto game under consideration. Players $A$ and $B$ compete over $n$ battlefields, whose valuations are given by $\{w_b\}_{b=1}^n$. In Stage 1, player $A$ decides how to deploy $P$ pre-allocated resources to the battlefields. Player $B$ observes the deployment. In Stage 2, the players simultaneously decide how to deploy their real-time resources $R_A$ and $R_B$, thus engaging in a General Lotto game with favoritism. 
    (b) A contour map of player $A$'s equilibrium payoff in the Stage 2 game under the optimal deployment of her pre-allocated resources $S$ in Stage 1. The dashed lines indicate level curves, i.e. the set of resource pairs $(P,R_A)\in\mbb{R}^2_+$ that achieve a given desired performance level $\Pi > 0$ (Theorem \ref{thm:level_set}). Here, we have normalized the battlefield values and player $B$'s budget such that $\sum_{b=1}^n w_b = 1$ and $qR_B = 1$. 
    (c) This table shows the relative effectiveness of pre-allocated to real-time resources, $P^\text{eq}/R_A$. Here, $P^\text{eq}$ is defined as the endowment $(P^\text{eq},0)$ (i.e. without real-time resources) that achieves the same performance $\Pi$ as the endowment $(0,R_A)$ (i.e. without pre-allocated resources) for a given $R_A$.  We find that real-time resources are at least twice as effective as pre-allocated resources, and can be arbitrarily more effective in certain settings (Corollary \ref{cor:twice_as_effective}). }
    \label{fig:main_figure}
\end{figure*}

% \subsection{Model}
The \emph{General Lotto game with pre-allocations} (GL-P) is a two-stage game with players $A$ and $B$, who compete over a set of $n$ battlefields, denoted as $\mcal{B}$. Each battlefield $b\in\mcal{B}$ is associated with a known valuation $w_b > 0$, which is common to both players. Player $A$ is endowed with a pre-allocated resource budget $P > 0$ and a real-time resource budget $R_A>0$. Player $B$ is endowed with a real-time resource budget $R_B > 0$, but no pre-allocated resources.\footnote{Recent computational advances (see, e.g., \cite{Vu_EC2021}) permit the study of the scenario where both players are endowed with pre-allocated resources.  In this work, we seek to provide \emph{analytical} characterizations of equilibrium payoffs, and, thus, consider the simpler, unilateral pre-allocation setting.}  The two stages are played as follows:

\smallskip \noindent -- \emph{Stage 1:} Player $A$ decides how to allocate her $P$ pre-allocated resources to the battlefields, i.e., it selects a vector $\bs{p} = (p_1,\ldots,p_n) \in \Delta_n(P) := \{\bs{p}' \in \mbb{R}_+^n : \|\bs{p}'\|_1 = P \}$. We term the vector $\bs{p}$ as player $A$'s \emph{pre-allocation profile}. No payoffs are derived in Stage 1, and $A$'s choice $\bs{p}$ becomes binding and common knowledge.

\smallskip \noindent -- \emph{Stage 2:} Players $A$ and $B$ compete in a simultaneous-move sub-game $G$ over $\mcal{B}$ with their real-time resource budgets $R_A$, $R_B$. Here, both players can randomly allocate these resources as long as their expenditure does not exceed their budgets in expectation. Specifically, a strategy for player $i \in \{A,B\}$ is an $n$-variate (cumulative) distribution $F_i$ over allocations $\bs{x}_i \in \mbb{R}_+^n$ that satisfies
\begin{equation}\label{eq:lotto_budget_constraint}
    \mbb{E}_{\bs{x}_i \sim F_i} \left[ \sum_{b\in\mcal{B}} x_{i,b} \right] \leq R_i.
\end{equation}

\noindent We use $\mcal{L}(R_i)$ to denote the set of all strategies $F_i$ that satisfy \eqref{eq:lotto_budget_constraint}.  Given that player $A$ chose $\bs{p}$ in Stage 1, the expected payoff to player $A$ is given by
\begin{equation}
    u_A(\bs{p},F_A,F_B) := \mbb{E}_{\substack{\bs{x}_A \sim F_A, \\ \bs{x}_B \sim F_B}}\left[ \sum_{b\in\mcal{B}} w_b \cdot I(x_{A,b} + p_b, qx_{B,b}) \right]
\end{equation}
where $I(a,b)=1$ if $a>b$, and $I(a,b)=0$ otherwise for any two numbers $a,b \in \mbb{R}_+$.\footnote{The tie-breaking rule (i.e., deciding who wins if $x_{A,b}+p_b=x_{B,b}$) can be assumed to be arbitrary, without affecting any of our results. This property is common in the General Lotto literature, see, e.g., \cite{Kovenock_2020,Vu_EC2021}.}  In words, player $B$ must overcome player $A$'s pre-allocated resources $p_b$ as well as player $A$'s allocation of real-time resources $x_{A,b}$ in order to win battlefield $b$.  The parameter $q > 0$ is the relative quality of player $B$'s real-time resources against player $A$'s resources. When $q < 1$ (resp. $q > 1$), they are less (resp. more) effective than player $A$'s resources. The payoff to player $B$ is $u_B(\bs{p},F_A,F_B) = W - u_A(\bs{p},F_A,F_B)$, where we denote $W=\sum_{b\in\mcal{B}} w_b$.

\smallskip Stages 1 and 2 of GL-P are illustrated in \Cref{fig:main_figure}a.  We denote an instance of GL-P as $\text{GL-P}(P,R_A,R_B,\bs{w})$, and note that the Stage 2 sub-game (i.e., the game with fixed pre-allocation profile) is an instance of the \emph{General Lotto game with favoritism} \cite{Vu_EC2021}.
% We note that the Stage 2 sub-game is an instance of the \emph{General Lotto game with favoritism} (GL-F) \cite{Vu_EC2021}, and denote the sub-game as $\text{GL-F}(\bs{p},R_A,R_B)$. 
For a given GL-P instance $G$, we define an equilibrium as any joint strategy profile $(\bs{p}^*,F_A^*,F_B^*) \in \Delta_n(P)\times\mcal{L}(R_A)\times\mcal{L}(R_B)$ that satisfies
\begin{equation}
    \begin{aligned}
        u_A(\bs{p}^*,F_A^*,F_B^*) &\geq u_A(\bs{p},F_A,F_B^*) \text{ and} \\ 
        u_B(\bs{p}^*,F_A^*,F_B^*) &\geq u_B(\bs{p}^*,F_A^*,F_B)
    \end{aligned}
\end{equation}
for any $\bs{p}\in\Delta_n(P)$, $F_A \in \mcal{L}(R_A)$ and $F_B\in\mcal{L}(R_B)$.  Notably, player $A$'s strategy consists of her deterministic pre-allocation profile $\bs{p}$ in Stage 1, as well as her randomized allocation of real-time resources $F_A$ in Stage 2.  It follows from the results in \cite{Vu_EC2021} that an equilibrium exists in any GL-P instance $G$, and that the equilibrium payoffs $\pi^*_i(G)=u_i(\bs{p}^*,F_A^*,F_B^*)$, $i\in\{A,B\}$, are unique.
% It was recently shown that an equilibrium exists in any GL-F instance $G$ \cite{Vu_EC2021}, and that the equilibrium payoffs $\pi_i(G) = u_i(F_A^*,F_B^*)$, $i\in\{A,B\}$, are unique. The same authors also formulate numerical methods to calculate an equilibrium of a given instance and the corresponding payoffs.  
For ease of notation, we will use $\pi^*_i(P,R_A,R_B)$, $i\in\{A,B\}$, to denote players' equilibrium payoffs in $G$ when the dependence on the vector $\bs{w}$ is clear.

\section{Main results}\label{sec:results}
% !TEX root = main.tex

In this section, we present our main result: the characterization of players' equilibrium payoffs in the GL-P game. %player $A$'s optimal deployment in Stage 1, as well as the final equilibrium payoff of the $\text{GL-F}$ game in Stage 2. 
We then use this result to derive an expression for the level sets of the function $\pi^*_A(P,R_A,R_B)$ in $(P,R_A)\in\mbb{R}_{\geq 0}\times\mbb{R}_{\geq 0}$, and to compare the relative effectiveness of pre-allocated and real-time resources.
% the optimal investment $(P^*,R_A^*)\in\mcal{I}(X_A)$ for the budget constraints $\mcal{I}(X_A)$ defined in \Cref{ex:linear_cost,ex:relative_effectiveness}.

% \subsection{Equilibrium characterization}
The result below 
% asserts that it is optimal for player $A$ to deploy its pre-allocated resources proportionally to the battlefield values $\bs{w}$ in Stage 1, and 
provides an explicit characterization of player $A$'s equilibrium payoff $\pi^*_A(P,R_A,R_B)$.  Note that player $B$'s equilibrium payoff is simply $\pi^*_B(P,R_A,R_B) = W - \pi^*_A(P,R_A,R_B)$.

\begin{theorem}\label{thm:equilibrium_characterization}
    Consider a GL-P game instance with $P, R_A, R_B > 0$, and $\bs{w}\in\mbb{R}^n_{++}$.  The following conditions characterize player $A$'s equilibrium payoff $\pi^*_A(P,R_A,R_B)$:
    
    \begin{enumerate}
        % \item There is an equilibrium $(\bs{p}^*,F_A^*,F_B^*)$ that satisfies
        % \begin{equation}
        %     \bs{p}^* := \frac{P}{W}\cdot\bs{w}.
        % \end{equation}
        \item If $qR_B < P$, or $qR_B \geq P$ and $R_A \geq \frac{2(qR_B-P)^2}{P+2(qR_B-P)}$, then $\pi^*_A(P,R_A,R_B)$ is
        \begin{equation} \label{eq:allcase1_payoff}
            W\cdot\left(1 - \frac{qR_B}{2R_A}\left(\frac{R_A + \sqrt{R_A(R_A+2P)}}{P+R_A+\sqrt{R_A(R_A+2P)}} \right)^2 \right).
        \end{equation}
        \item Otherwise, $\pi^*_A(P,R_A,R_B)$ is
        \begin{equation} \label{eq:allcase2_payoff}
            W\cdot\frac{R_A}{2(qR_B-P)}.
        \end{equation}
    \end{enumerate}
\end{theorem}

\noindent The derivation of the above result is challenging because explicit expressions for the players' payoffs in the Stage 2 sub-game are generally not attainable for arbitrary $\bs{p} \in \Delta_n(P)$. Moreover, these payoffs are not generally concave. Our approach is to show that for any $\bs{p} \neq \bs{p}^*$, the payoff is nondecreasing in the direction pointing to $\bs{p}^*$. The full proof is given in Appendix \ref{sec:thm_proof}, and relies on methods developed in \cite{Vu_EC2021}. These details are given in Appendix \ref{sec:Vu_method}. 

As a consequence of our main result, we are able to characterize expressions for the level curves of the function $\pi_A^*(P,R_A,R_B)$. That is, for a desired performance level $\Pi \geq 0$ and fixed $R_B$, we provide the set of all pairs $(P,R_A)$ such that $\pi^*_A(P,R_A,R_B) = \Pi$.

% One approach to solving the above optimization problem is to characterize the level sets of the function $\pi^*_A(P,R_A,R_B)$ for given $R_B>0$ and $\bs{w}\in\mbb{R}^n_{++}$, i.e., derive an expression for the set $\{(P,R_A)\in\mbb{R}^2_+ \text{ s.t. } \pi^*_A(P,R_A,R_B)=\Pi\}$ for some target payoff $\Pi\geq0$. 

\begin{theorem}\label{thm:level_set}
    Given any $R_B>0$ and $\bs{w}\in\mbb{R}^n_{++}$, fix a desired performance level $\Pi \in [0,W]$.  The set of all pairs $(P,R_A) \in \mbb{R}_+^2$ that satisfy $\pi^*_A(P,R_A,R_B) = \Pi$ is given by the following conditions:  
    
    \smallskip\noindent If $0\leq \Pi < \frac{W}{2}$, then
    \begin{equation} \label{eq:level_curve_1and2}
        R_A = \!\!\begin{cases}
            \frac{2\Pi}{W} (qR_B-P) &\, \text{if } P \in \Big[ 0, \frac{(W-2\Pi)qR_B}{W-\Pi} \Big) \\
            \!\!\frac{(qR_BW-(W-\Pi)P)^2}{2qR_B(W-\Pi)W} &\, \text{if } P \in \Big[\frac{(W-2\Pi)qR_B}{W-\Pi},\frac{WqR_B}{W-\Pi}\Big]
        \end{cases}
    \end{equation}
    
    \smallskip\noindent If $\frac{W}{2} \leq \Pi \leq W$, then
    \begin{equation}
        R_A = \frac{(qR_BW-(W-\Pi)P)^2}{2qR_B(W-\Pi)W}, \text{ if } P \in \Big[0,\frac{WqR_B}{W-\Pi}\Big]
    \end{equation}
    If $P> \frac{W qR_B}{W-\Pi}$, then $\pi^*_A(P,R_A,R_B)>\Pi$ for any $R_A\geq0$.
\end{theorem}

\noindent We plot the surface $\pi^*_A(P,R_A,R_B)$ for $(P,R_A) \in \mbb{R}_+^2$ as well as the level curves corresponding to $\Pi\in\{0.250,0.500,0.625,0.750,0.875\}$ in Figure \ref{fig:main_figure}b. Notably, for any $\Pi \in (0,W)$, the level curve $R_A^\Pi(P)$ is strictly decreasing and convex in $P \in [0,\frac{qR_BW}{W-\Pi}]$, where we use $R_A^\Pi(P)$ to explicitly note the dependence on $\Pi$. Hence, the function $\pi_A(P,R_A,R_B)$ is quasi-concave in $(P,R_A)$.

%%

% \subsection{Relative effectiveness of pre-allocated resources}

% Here we consider the budget constraint from \Cref{ex:relative_effectiveness}.  
We can use the result in \Cref{thm:level_set} to obtain an expression for the relative effectiveness of pre-allocated and real-time resources when these are deployed in isolation.  In the following corollary, we provide this expression, and observe that real-time resources are at least twice as valuable as pre-allocated resources, and can be arbitrarily more valuable in specific settings:

\begin{corollary} \label{cor:twice_as_effective}
    For given $R_A, R_B>0$, the unique value $\eq P>0$ such that $\pi^*_A(\eq P,0,R_B)=\pi^*_A(0,R_A,R_B)$ is characterized by the following expression:
    \begin{equation}
        \eq P =
        \begin{cases}
            2R_A & \quad \text{if } R_A > qR_B, \\
            \frac{2(qR_B)^2}{2qR_B-R_A} & \quad \text{if } R_A \leq qR_B.
        \end{cases}
    \end{equation}
    Notably, the ratio $\eq P/R_A$ is lower-bounded by $2$, and $\eq P/R_A\to \infty$ as $R_A\to 0^+$.
\end{corollary}

\noindent The table in \Cref{fig:main_figure}c compares the relative effectiveness of pre-allocated and real-time resources corresponding with the performance levels considered in \Cref{fig:main_figure}b.

% We utilize analytical methods developed in \cite{Vu_EC2021} to derive the analytical expression above, as we outline in Appendix \ref{sec:Vu_method}. % In Appendix \ref{sec:thm_proof}, we provide our derivation of the expression and identify its structural properties in order to solve \eqref{eq:reduced_investment_problem}.

%%

\section{Optimal investment decisions}

In this section, we consider a scenario where player $A$ has an opportunity to make an investment decision regarding its resource endowments. That is, the pair $(P,R_A)\in\mbb{R}^2_+$ is a strategic choice made by $A$ before the game $\text{GL-P}(P,R_A,R_B,\bs{w})$ is played. Given a monetary budget $X_A > 0$ for player $A$, any pair $(P,R_A)$ must belong to the following set of feasible investments:
\begin{equation}\label{eq:linear_cost_constraint}
        \mcal{I}(X_A) := \{(P,R_A) : R_A + cP \leq X_A\}
\end{equation}
where $c\geq 0$ is the per-unit cost for purchasing pre-allocated resources, and we assume the per-unit cost for purchasing real-time resources is 1 without loss of generality. We are interested in characterizing player $A$'s optimal investment subject to the above cost constraint, and given player $B$'s resource endowment $R_B>0$. This is formulated as the following optimization problem:
\begin{equation}\label{eq:optimal_investment_problem}
    \pi_A^\mathrm{opt} := \max_{(P,R_A) \in \mcal{I}(X_A)} \pi_A^*(P,R_A,R_B).
\end{equation}

% \begin{example}[Linear cost constraint] \label{ex:linear_cost}
%     Consider the setting where player $A$'s choice in Stage 1 must satisfy a linear cost constraint, i.e., 
%     \begin{equation} \label{eq:linear_cost_constraint}
%         \mcal{I}(X_A) = \{(P,R_A)\in\mbb{R}^2_+ \text{ s.t. } R_A + cP \leq X_A\},
%     \end{equation}

%     \noindent where $c\geq 0$ represents the relative cost of pre-allocated resources to real-time resources.  
% \end{example}

% \begin{example}[Relative effectiveness of pre-allocated resources] \label{ex:relative_effectiveness}
%     Consider the setting where player $A$'s choice in Stage 1 must satisfy a constraint of the form $\mcal{I}(X_A) = \{(X_A/c,0), (0,X_A)\}$, where $c\geq 0$ represents the relative cost of pre-allocated to real-time resources.  By analyzing this constraint, we study the relative effectiveness of pre-allocated and real-time resources when deployed in isolation.
% \end{example}

In the result below, we derive the complete solution to the optimal investment problem \eqref{eq:optimal_investment_problem}. 

\begin{corollary}\label{cor:investment}
    Fix a monetary budget $X_A>0$, relative per-unit cost $c>0$, and $R_B>0$ real-time resources for player $B$.  Then, player $A$'s optimal investment in pre-allocated resources for the optimization problem in \eqref{eq:optimal_investment_problem} under the linear cost constraint in \eqref{eq:linear_cost_constraint} is
    \begin{equation}
        P^* = 
        \begin{cases}
            (1 - \frac{c}{2-c})\frac{X_A}{c}, &\text{if } c < t \\
            \in [0,(1 - \frac{c}{2-c})\frac{X_A}{c}], &\text{if } c = t \\
            0, &\text{if } c > t
        \end{cases}.
    \end{equation}
    where $t := \min\{1,\frac{X_A}{qR_B}\}$. The optimal investment in real-time resources is $R_A^* = X_A - cP^*$. The resulting payoff $\pi_A^\mathrm{opt}$ to player $A$ is given by
    \begin{equation}
        W\cdot
        \begin{cases}
            1 - \frac{qR_B}{2X_A}c(2-c), &\text{if } c < t \\
            1 - \frac{qR_B}{2X_A}, &\text{if } c \geq t \text{ and } \frac{X_A}{qR_B} \geq 1 \\
            \frac{X_A}{2qR_B}, &\text{if } c \geq t \text{ and } \frac{X_A}{qR_B} < 1 
        \end{cases}.
    \end{equation}
\end{corollary}
The above solution is obtained by leveraging the level set characterization from \Cref{thm:level_set}, and the fact that the level sets are strictly decreasing and convex for $\Pi\in(0,W)$.  We omit details of the proof for space considerations. A visual illustration of how the optimal investments are determined is shown in Figure \ref{fig:linear_cost}. The budget constraint $\mcal{I}(X_A)$ is a line segment in $\mbb{R}_+^2$, and we thus seek the level curve that lies tangent to the segment. In cases where the cost $c$ is sufficiently high, no level curve lies tangent to $\mcal{I}(X_A)$, and, thus, player $A$ invests all of her budget in real-time resources.

\begin{figure}[t!]
    \centering
    \includegraphics[width=0.45\textwidth,trim={0 7.75cm 29.25cm 0}]{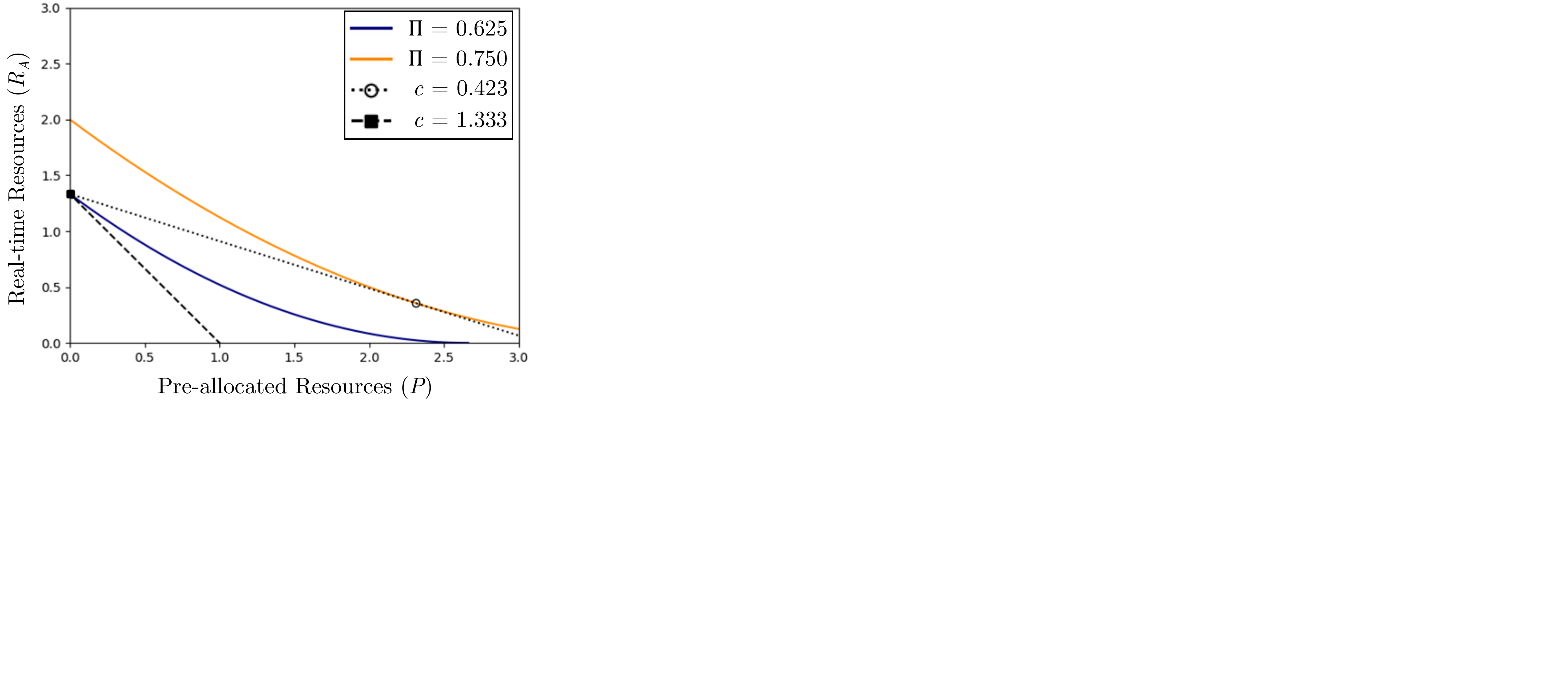}
    \caption{The optimal investment $(P^*,R^*_A)\in\mbb{R}^2_+$ subject to the linear cost constraint in \eqref{eq:linear_cost_constraint}.  Here, we consider the optimal investment problem when $X_A=4/3$, $qR_B=W=1$, and $c\in\{0.423,1.333\}$.  Observe that the set of feasible investments $\mcal{I}(X_A)$ is the line segment connecting $(0,X_A)$ and $(X_A/c,0)$. The optimal investment lies on the level curve tangent to this line segment. For example, when $c=0.423$, the optimal investment is $(2.309,0.357)$ (unfilled circle), as $\mcal{I}(X_A)$ (dotted, black line) is tangent to the level curve with $\Pi=0.750$ (solid, orange line). For sufficiently high cost $c$, $\mcal{I}(X_A)$ will not be tangent to any level curve, and the optimal investment is $(0,X_A)$.  For example, when $c=1.333$, observe that $\mcal{I}(X_A)$ (dashed, black line) is not tangent even to the level curve with $\Pi=0.625$ (solid, blue line), and the optimal investment is $(0,4/3)$ (filled square).}
    \label{fig:linear_cost}
\end{figure}

\section{Conclusion}

In this manuscript, we studied the strategic role of pre-allocations in competitive interactions under a two-stage General Lotto game model. We identified an explicit expression for the set of pre-allocated and real-time budget pairs that correspond with a given desired performance.  We then used this explicit expression to derive the optimal dynamic investment strategy under a given linear cost constraint, and to compare the relative effectiveness of pre-allocated and real-time resources when deployed in isolation.  Exciting directions for future work include studying the strategic outcomes (i.e., equilibria) when both players can make pre-allocations, and introducing heterogeneities in players' battlefield valuations and resource effectiveness to the model.

\bibliographystyle{IEEEtran}
\bibliography{sources}

\appendix
% !TEX root = main.tex

\subsection{Method to derive equilibria of second-stage subgame}\label{sec:Vu_method}

The recent work of Vu and Loiseau \cite{Vu_EC2021} provides a general method to derive an equilibrium of the second stage subgame from the GL-P game, which is termed a General Lotto game with favoritism (GL-F). In a GL-F game, the pre-allocation vector $\bs{p}$ is an exogenous parameter. We denote an instance of as $\text{GL-F}(\bs{p},R_A,R_B)$. The method to calculate an equilibrium involves solving the following system\footnote{The problem setting considered in their method is more general, admitting possibly negative pre-allocations $p_b < 0$ (i.e. favoring player $B$), asymmetries in players' battlefield valuations $w_b>0$, and different resource effectiveness parameters $q_b$ for each battlefield. However, exact closed-form solutions under heterogeneous values $\bs{w}$, arbitrary pre-allocations $\bs{p}$, and effectiveness parameters $q_b$ are generally unattainable.} of two equations for two unknowns $(\kappa_A,\kappa_B) \in \mbb{R}_{++}^2$: 
\begin{equation}\label{eq:SOE}
    \begin{aligned}
        R_A \!=\! \sum_{b=1}^n \frac{[h_b(\kappa_A,\kappa_B) - p_b]^2}{2qw_b\kappa_B}, \ 
        R_B \!= \sum_{b=1}^n \frac{h_b^2(\kappa_A,\kappa_B) - p_b^2}{2qw_b\kappa_A}
    \end{aligned}
\end{equation}
where $h_b(\kappa_A,\kappa_B) := \min\{ q w_b\kappa_B, w_b\kappa_A + p_b\}$ for $b \in \mcal{B}$. The above equations correspond to the budget constraint \eqref{eq:lotto_budget_constraint} for both players. There always exists a solution $(\kappa_A^*,\kappa_B^*) \in \mbb{R}_{++}^2$ to this system \cite{Vu_EC2021}, and corresponds to the following equilibrium payoffs. 

\begin{lemma}[Adapted from \cite{Vu_EC2021}]\label{lem:SOE}
    Suppose $(\kappa_A^*,\kappa_B^*) \in \mbb{R}_{++}^2$ solves \eqref{eq:SOE}.  Let $\mcal{B}_1 := \{b\in\mcal{B} : h_b(\kappa_A^*,\kappa_B^*) = qw_b\kappa_B\}$ and $\mcal{B}_2 = \mcal{B}\backslash \mcal{B}_1$.
    Then there is a corresponding equilibrium $(F_A^*,F_B^*)$ of the game $\text{GL-F}(\bs{p},R_A,R_B)$ where player $A$'s equilibrium payoff is given by 
    \begin{equation} \label{eq:playerA_payoff}
        \begin{aligned}
            \pi_A(\bs{p},R_A,R_B) &:= \sum_{b\in\mcal{B}_1} \!w_b\!\left[1 - \frac{q\kappa_B^*}{2\kappa_A^*}\left(1 - \frac{p_i^2}{(qw_b\kappa_B)^2} \right) \right] \\
            %&= \sum_{b\in\mcal{B}_1} w_b\left[1 - \frac{qw_b\kappa_B^* - p_b}{w_b\kappa_A^*} + \frac{(qw_b\kappa_B^*-p_b)^2}{2w_b^2\kappa_A^*\kappa_B^*} \right] \\
            &\quad + \sum_{b\in\mcal{B}_2} w_b \frac{\kappa_A^*}{2q\kappa_B^*}
        \end{aligned}
    \end{equation}
    and the equilibrium payoff to player $B$ is $\pi_B(\bs{p},R_A,R_B) = W - \pi_A(\bs{p},R_A,R_B)$.
    
    % The equilibrium strategies $F_A^*$, $F_B^*$ for players $A$ and $B$ are given as follows. For all $b \in \mcal{B}_1$, the marginal (cumulative) distributions are
    % \begin{equation}
    %     \begin{aligned}
    %         F_{A,b}^* &= \frac{p_b}{qw_b \kappa_B}\delta_0 + (1- \frac{p_b}{qw_b \kappa_B}) \text{Unif}(0,qw_b\kappa_B - p_b) \\
    %         F_{B,b}^* &= (1 - \frac{qw_b\kappa_B - p_b}{w_b\kappa_A})\delta_0 + \frac{qw_b\kappa_B - p_b}{w_b\kappa_A} \text{Unif}(\frac{p_b}{q},w_b\kappa_B)
    %     \end{aligned}
    % \end{equation}
    % For all $b \in \mcal{B}_2$, the univariate marginal distributions are
    % \begin{equation}
    %     \begin{aligned}
    %         F_{A,b}^* &= (1 - \frac{\kappa_A}{q \kappa_B})\delta_0 + \frac{\kappa_A}{q \kappa_B} \text{Unif}(0,w_b\kappa_A) \\
    %         F_{B,b}^* &= \text{Unif}(\frac{p_b}{q},\frac{w_b\kappa_A + p_b}{q})
    %     \end{aligned}.
    % \end{equation}
    % Here, we use $\delta_0$ to denote the CDF of a unit probability mass at $x = 0$, and $\text{Unif}(a,b)$ to denote the CDF of a uniform distribution on the interval $[a,b]$. 
\end{lemma}
The equilibrium strategies are characterized by marginal distributions detailed in \cite{Vu_EC2021}.
%%

% % % Proof of optimal pre-allocation % % %
\subsection{Proof of Theorem \ref{thm:equilibrium_characterization}}\label{sec:thm_proof}

The proof follows two parts:  In Part 1, we establish that, for given $P,R_A,R_B>0$ and $\bs{w}\in\mbb{R}^n_{++}$, $\bs{p}^* = \frac{P}{W}\bs{w}$ is player $A$'s optimal pre-allocation profile in Stage 1 of GL-P.  Then, in Part 2, we derive an explicit expression for player $A$'s payoff in Stage 2 under the optimal pre-allocation profile $\bs{p}^*$ derived in Part 1.  Throughout the proof, we use $\pi_i(\bs{p},R_A,R_B)$, $i\in\{A,B\}$, to denote the players' payoffs in the Stage 2 sub-game for fixed pre-allocation profile $\bs{p}\in\Delta_n(P)$.  Recall that the Stage 2 sub-game amounts to a General Lotto game with favoritism $\text{GL-F}(\bs{p},R_A,R_B)$.

\smallskip\noindent-- {\bf Part 1:} The proof amounts to showing that $\bs{p}^* = \frac{P}{W}\bs{w}$ is a global maximizer of player $A$'s equilibrium payoff $\pi^*_A(\bs{p},R_A,R_B)$ for $\bs{p} \in \Delta_n(P)$. For the following analysis, we define $\mbb{T}_n := \{\bs{z} \in \mbb{R}^n : \sum_{b=1}^n z_b = 0$ as the tangent space of $\Delta_n(P)$. The lemma below first establishes that $\bs{p}^*$ is a local maximizer when either $\mcal{B}_1=\mcal{B}$ or $\mcal{B}_2=\mcal{B}$.

\begin{lemma} \label{lem:local_maximizer}
    The pre-allocation $\bs{p}^* = \frac{P}{W}\bs{w}$ is a local maximizer of $\pi_A(\bs{p},R_A,R_B)$ over $\bs{p} \in \Delta_n(P)$, for any $P,R_A,R_B > 0$.
\end{lemma}
\begin{proof}
    From \Cref{lem:SOE} and the definition of $h_b(\kappa_A,\kappa_B)$ in \Cref{sec:Vu_method}, we observe that the solution to \eqref{eq:SOE} under the pre-allocation $s^*$ is always in one of two completely symmetric cases: 1) $\mcal{B}_1 = \mcal{B}$; or 2) $\mcal{B}_2 = \mcal{B}$.  Thus, we need to show $s^*$ is a local maximizer in both cases.
    % 1) $\mcal{B}_1 = \mcal{B}$ iff $qR_B - P < 0$ or $qR_B - P \geq 0$ and $R_A \leq \frac{2(qR_B-P)^2}{P + (qR_B-P)}$; or 2) $\mcal{B}_2 = \mcal{B}$, otherwise. We thus need to show $s^*$ is a local maximium in both cases. 

    \smallskip\noindent\underline{{\bf Case 1 ($\mcal{B}_1 = \mcal{B}$):}} For $\bs{p} \in \Delta_n(P)$, the system \eqref{eq:SOE} is written
    \begin{equation}\label{eq:case22_SOE}
        \begin{aligned}
            &R_A = \sum_{b=1}^n \frac{(qw_b\kappa_B - p_b)^2}{2qw_b\kappa_B} \text{ and } R_B = \sum_{b=1}^n \frac{(qw_b\kappa_B)^2-p_b^2}{2qw_b\kappa_A} \\
            &\text{where } 0 < qw_b\kappa_B - p_b \leq \kappa_A \text{ holds } \forall b \in \mcal{B}.
        \end{aligned}
    \end{equation}
    It yields an algebraic solution
    \begin{equation} \label{eq:case1solution_closedform}
        \begin{aligned}
            q\kappa_B^* &= \frac{1}{W}\left[P+R_A + \sqrt{(P+R_A)^2 - W\|\bs{p}\|_{\bs{w}}^2} \right] \\ 
            \kappa_A^* &= \frac{(P+R_A)q\kappa_B^* - \|\bs{p}\|_{\bs{w}}^2}{qR_B}.
        \end{aligned}
    \end{equation}
    where $\|\bs{p}\|_{\bs{w}}^2 = \sum_{b=1}^n \frac{p_b^2}{w_b}$. This solution needs to satisfy the set of conditions $0 < qw_b\kappa_B - p_b \leq \kappa_A \ \forall b \in \mcal{B}$, but the explicit characterization of these conditions is not needed to show that $s^*$ is a local maximum. Indeed, first observe that the expression for $q\kappa_B^*$ is required to be real-valued, which we can write as the condition
    \begin{equation}
        \bs{p} \in R^{(1n)} := \left\{\bs{p} \in \Delta_n(P) : \|\bs{p}\|_{\bs{w}}^2 < \frac{(P+R_A)^2}{W} \right\}.
    \end{equation}
    We thus have a region $R^{(1n)}$ for which the expression of player $A$'s equilibrium payoff (derived using Lemma \ref{lem:SOE}) is well-defined:
    \begin{equation}\label{eq:case22_payoff}
        \pi_A^{(1n)}(\bs{p}) := W\left(\!1\! - \!\frac{qR_B}{f(\|\bs{p}\|_{\bs{w}})}\!\left(\!1 - \frac{W\|\bs{p}\|_{\bs{w}}^2}{(P\!+\!R_A\!+\!f(\|\bs{p}\|_{\bs{w}}))^2} \right) \!\right)
        % 2 + \frac{X_B}{(P+X_A)^2+4p_1(P-p_1)-2P^2}\left(\frac{4p_1(P-p_1) - 2P^2}{(P+X_A+\sqrt{(P+X_A)^2+4p_1(P-p_1)-2P^2})^2} - 1 \right)
    \end{equation}
    where $f(\|\bs{p}\|_{\bs{w}}) := \sqrt{(P+R_A)^2 - W\|\bs{p}\|_{\bs{w}}^2}$. The partial derivatives are calculated to be
    \begin{equation}
        \frac{\partial \pi_A^{(1n)}}{\partial p_b}(\bs{p}) = \frac{p_b}{w_b} \cdot\frac{2W^2qR_B }{f(\|\bs{p}\|_{\bs{w}})(P+R_A+f(\|\bs{p}\|_{\bs{w}}))^2}
        %\nabla \pi_A(\bs{p}) = \left(\frac{\partial \pi_A}{\partial p_1}, \ldots, \frac{\partial \pi_A}{\partial p_n} \right)^\top
    \end{equation}
    A critical point of $\pi_A^{(1n)}$ must satisfy $\bs{z}^\top \nabla \pi_A^{(1n)}(\bs{p})=0$ for any $\bs{z} \in \mbb{T}_n$. Indeed for any $\bs{p} \in R^{(1n)}$, we calculate
    \begin{equation}
        \begin{aligned}
            % (\bs{p} - \frac{P}{W}\bs{w})^\top \nabla \pi_A^{(1n)}(\bs{p}) &= g(\|\bs{p}\|_{\bs{w}})\cdot \sum_{b=1}^n \left(p_b - \frac{P}{W}w_b \right)\frac{p_b}{w_b} \\
            (\bs{p} - \frac{P}{W}\bs{w})^\top \nabla \pi_A^{(1n)}(\bs{p}) &= g(\|\bs{p}\|_{\bs{w}})\cdot \left(\|\bs{p}\|_{\bs{w}}^2 - \frac{P^2}{W} \right) \\
            &\geq 0
        \end{aligned}
    \end{equation}
    where $g(\|\bs{p}\|_{\bs{w}}):=\frac{2W^2qR_B }{f(\|\bs{p}\|_{\bs{w}})(P+R_A+f(\|\bs{p}\|_{\bs{w}}))^2} > 0$ for any $\bs{p} \in R^{(1n)}$. The inequality above is met with equality if and only if $\bs{p} = \bs{p}^*$. This is due to the fact that $\min_{\bs{p} \in \Delta_n(P)} \|\bs{p}\|_{\bs{w}}^2 = \|\bs{p}^*\|_{\bs{w}}^2 = \frac{P^2}{W}$. Thus, $\bs{p}^*$ is the unique maximizer of $\pi_A^{(1n)}(\bs{p})$ on $R^{(1n)}$.
    
    \smallskip\noindent\underline{{\bf Case 2 ($\mcal{B}_2 = \mcal{B}$):}}
    For $\bs{p}\in\Delta_n(P)$, the system is written as
    \[ R_A = \sum^n_{b=1} \frac{(w_b \kappa_A)^2}{2q w_b \kappa_B} \text{ and } 
       R_B = \sum^n_{b=1} \frac{(w_b \kappa_A-p_b)^2-(p_b)^2}{2q w_b \kappa_A}, \]
    where $q w_b \kappa_B - p_b > w_b \kappa_A$ holds for all $b\in\mcal{B}$.  This readily yields the algebraic solution:
    \begin{equation} \label{eq:case2solution_closedform}
        q\kappa_B^* = \frac{2}{W} \frac{(q R_B-P)^2}{R_A} \text{ and } 
        \kappa_A^* = \frac{2}{W} (q R_B-P).
    \end{equation}
    For this solution to be valid, the following conditions are required:
    
    % \begin{itemize}
        \noindent $\bullet$ $\kappa_A^*,q\kappa_B^* \in \mbb{R}_{++}$:  This requires that $qR_B-P>0$.
        
        \noindent $\bullet$ $q w_b \kappa_B^* - p_b > w_b \kappa_A^*$ for all $b\in\mcal{B}$:  This requires that
        \[ \frac{2}{W} \frac{(q R_B-P)^2}{R_A} - \frac{2}{W} (q R_B-P) - \max_{b} \{\frac{p_b}{w_b}\} > 0. \]
        The left-hand side is quadratic in $q R_B-P$, and, thus, requires either
        \[\begin{aligned}
            q R_B-P &< \frac{2/W - \sqrt{\frac{4}{W^2}+4 \max_b\{\frac{p_b}{w_b}\}\frac{2}{W R_A}}}{4/(W R_A)}, \text{ or }\\
            q R_B-P &> \frac{2/W + \sqrt{\frac{4}{W^2}+4 \max_b\{\frac{p_b}{w_b}\}\frac{2}{W R_A}}}{4/(W R_A)}.
        \end{aligned} \]
        The former cannot hold since the numerator on the right-hand side is strictly negative, but $\kappa_A^*,q\kappa_B^* \in \mbb{R}_{++}$ requires $q R_B-P>0$.  Thus, the latter must hold, which simplifies to the condition
        \begin{equation} \label{eq:case3_mostrestrictive}
            q R_B-P > \frac{R_A}{2} \left[ 1 + \sqrt{1 + \frac{2W}{R_A}\max_b\{\frac{p_b}{w_b}\}} \right].
        \end{equation}
    % \end{itemize}
    Clearly, \eqref{eq:case3_mostrestrictive} is more restrictive than $q R_B-P>0$, and, thus, dictates the boundary of Case 2.
    
    For any $\bs{p}\in\Delta_n(P)$ such that all battlefields are in Case 2, the expression for player $A$'s payoff in \eqref{eq:playerA_payoff} simplifies to
    \[ \pi_A(\bs{p},R_A,R_B) = \sum^n_{b=1} w_b \frac{\kappa_A^*}{2q \kappa_B^*} 
        = \frac{W}{2} \frac{R_A}{q R_B-P}, \]
    where we use the expression for $q\kappa_B^*$ and $\kappa_A^*$ in \eqref{eq:case2solution_closedform}.  Observe that player $A$'s payoff is constant in the quantity $\bs{p}$.  Thus, for any $\bs{p}$ that satisfies \eqref{eq:case3_mostrestrictive}, it holds that all battlefields are in Case 2, and that player $A$'s payoff is the above.  We conclude the proof noting that, for given quantities $R_A$ and $P$, if there exists any $\bs{p}\in\Delta_n(P)$ such that \eqref{eq:case3_mostrestrictive} is satisfied, then $\bs{p}^* = \bs{w} \cdot (P/W)$ must also satisfy \eqref{eq:case3_mostrestrictive}, since $||\bs{p}||_\infty \geq ||\bs{p}^*||_\infty$ and the right-hand side in \eqref{eq:case3_mostrestrictive} is increasing in $||\bs{p}||_\infty$.
\end{proof}

Next, we prove that the function $\pi_A(\bs{p},R_A,R_B)$ is maximized by $\bs{p}^*=\frac{P}{W}\bs{w}$. We showed in Lemma \ref{lem:local_maximizer} that $\bs{p}^*$ is a local maximizer over $\bs{p}\in\Delta_n(P)$ when either $\mcal{B}_1=\mcal{B}$ or $\mcal{B}_2=\mcal{B}$.  It remains to be shown that player $A$ cannot achieve a higher payoff for $\bs{p}\in\Delta_n(P)$ that results in both sets $\mcal{B}_1$ and $\mcal{B}_2$ being nonempty.  Throughout the proof, we will use the short-hand notation $W_j=\sum_{b\in\mcal{B}_j} w_b$, $P_j=\sum_{b\in\mcal{B}_j} p_b$ and $\bs{p}_j=(p_b)_{b\in\mcal{B}_j}$, $j=1,2$, for conciseness.

For $\bs{p}\in\Delta_n(P)$, we have that
\begin{equation*}
\begin{aligned}
    X_A &= \sum_{b\in\mcal{B}_1} \frac{(q w_b \kappa_B-p_b)^2}{2 q w_b \kappa_B} 
            + \sum_{b\in\mcal{B}_2} \frac{(w_b \kappa_A)^2}{2 q w_b \kappa_B}, \\
    X_B &= \sum_{b\in\mcal{B}_1} \frac{(q w_b \kappa_B)^2-(p_b)^2}{2 q w_b \kappa_A} 
            + \sum_{b\in\mcal{B}_2} \frac{(w_b \kappa_A+p_b)^2-(p_b)^2}{2 q w_b \kappa_A},
\end{aligned}
\end{equation*}
where $0<qw_b\kappa_B-p_b\leq w_b\kappa_A$ holds for all $b\in\mcal{B}_1$, and $q w_b \kappa_B-p_b>w_b\kappa_A$ holds for all $b\in\mcal{B}_2$.  The system of equations readily gives the expression:
\begin{equation} \label{eq:mixedcases_systemeq}
\begin{aligned}
    W_1 (q\kappa_B)^2 + W_2 (\kappa_A)^2 &= 2q\kappa_B (X_A+P_2) - ||\bs{p}_1||^2_w \\
    &= 2\kappa_A (q X_B-P_2) + ||\bs{p}_1||^2_w,
\end{aligned}
\end{equation}
where recall that $||\bs{p}_1||^2_w = \sum_{b\in\mcal{B}_2} [(p_b)^2/w_b]$.  The solution to the above system of equations is
\begin{equation} \label{eq:mixedcases_solution}
\begin{aligned}
    q\kappa^*_B &= \frac{C_1 H_2 \pm \sqrt{(C_2)^2H_1H_2}}{W_1 (C_2)^2 + W_2 (C_1)^2}, \\
    \kappa^*_A  &= \frac{C_2 H_1 \pm \sqrt{(C_1)^2H_1H_2}}{W_1 (C_2)^2 + W_2 (C_1)^2},
\end{aligned}
\end{equation}
where we introduce the short-hand notation $C_1=R_A+P_1$, $C_2=qR_B-P_2$, $H_1=(C_1)^2-W_1||\bs{p}_1||^2_w$ and $H_2=(C_2)^2+W_2||\bs{p}_1||^2_w$, for conciseness.  We consider only the scenario where $\pm=+$ in \eqref{eq:mixedcases_solution}, since the expression for $\kappa^*_A$ is strictly negative when $\pm=-$.  Simply observe that $C_1>0$, $(C_1)^2 > H_1$, $0<(C_2)^2 < H_2$ and, thus, that either (i) $H_1>0$, $C_2>0$ and $0<C_2H_1 < C_1\sqrt{H_1H_2}$, (ii) $H_1<0$, $C_2<0$ and $0< C_2H_1 = |C_2| |H_1| < C_1 \sqrt{|H_1| |H_2|}$, or (iii) only one of $H_1$ or $C_2$ is negative, in which case $C_2H_1 < 0$.

Substituting \eqref{eq:mixedcases_solution} into \eqref{eq:playerA_payoff} and simplifying, we obtain
\begin{equation}
\begin{aligned}
    \pi_A(\bs{p},R_A,R_B) %&= \sum_{b\in\mcal{B}_1} \left[ w_b - w_b \frac{q\kappa^*_B}{\kappa^*_A} + \frac{p_b}{\kappa^*_A} + w_b \frac{(q \kappa^*_B - (p_b/w_b))^2}{2\kappa^*_A q \kappa^*_B} \right] \\ &+ \sum_{b\in\mcal{B}_2} \left[ \frac{w_b \kappa^*_A}{2q\kappa^*_B} \right] \\
    % &= W_1 \Big[ 1-\frac{q\kappa^*_B}{2\kappa^*_A} \Big] +\frac{||\bs{p}_1||^2_w}{2\kappa^*_Aq\kappa^*_B} + W_2 \frac{\kappa^*_A}{2q\kappa^*_B} \\
    &= W_1 + \frac{\sqrt{H_1H_2}-C_1C_2}{||\bs{p}_1||^2_w},
\end{aligned}
\end{equation}

\noindent and the partial derivatives of $\pi_A(\bs{p},R_A,R_B)$ with respect to $p_b$ for $b\in\mcal{B}_1$ and $b\in\mcal{B}_2$, respectively, are:
\begin{equation} \label{eq:mixed_partials}
\begin{aligned}
    \left. \frac{\partial}{\partial p_b} \pi_A \right|_{b\in\mcal{B}_1} &= 
        \frac{-p_b/w_b}{(||\bs{p}_1||^2_w)^2\sqrt{H_1H_2}} (C_1\sqrt{H_2}-C_2\sqrt{H_1})^2 \\
        &\quad +\frac{1}{||\bs{p}_1||^2_w\sqrt{H_1}}(C_1\sqrt{H_2}-C_2\sqrt{H_1}) \\
    \left. \frac{\partial}{\partial p_b} \pi_A \right|_{b\in\mcal{B}_2} &= 
        \frac{1}{||\bs{p}_1||^2_w\sqrt{H_2}}(C_1\sqrt{H_2}-C_2\sqrt{H_1}).
\end{aligned}
\end{equation}

We first consider critical points $\bs{p}$ strictly in the interior of $\Delta_n(P)$, and resolve the points on the boundary later.  One necessary condition for a critical point is that $\partial \pi_A/(\partial p_b) - \partial \pi_A/(\partial p_c) = 0$ for all $b\in\mcal{B}_1$ and $c\in\mcal{B}_2$.  Firstly, observe that $C_1 > \sqrt{H_1}$ and $\sqrt{H_2}>C_2$, and, thus, it must be that $C_1\sqrt{H_2}-C_2\sqrt{H_1}>0$.  We can thus divide the expression $\partial \pi_A/(\partial p_b) = \partial \pi_A/(\partial p_c)$ on both sides by $C_1\sqrt{H_2}-C_2\sqrt{H_1}$ and rearrange to obtain
% \[ \frac{-p_b/w_b}{(||\bs{p}_1||^2_w)^2\sqrt{H_1H_2}} (C_1\sqrt{H_2}-C_2\sqrt{H_1})
        % +\frac{1}{||\bs{p}_1||^2_w\sqrt{H_1}} = \frac{1}{||\bs{p}_1||^2_w\sqrt{H_2}}. \]
\[ (p_b/w_b) (C_1\sqrt{H_2}-C_2\sqrt{H_1})
        = ||\bs{p}_1||^2_w (\sqrt{H_2} - \sqrt{H_1}) > 0. \]
Observe that the left-hand side is strictly greater than zero, and, thus, the right-hand side must be as well.  This immediately requires $\sqrt{H_2}-\sqrt{H_1}>0$, since $||\bs{p}_1||^2_w>0$.  Re-arranging the above expression, note that we also require
\[ \sqrt{H_1} [C_2 (p_b/w_b) - ||\bs{p}_1||^2_w] = \sqrt{H_2} [ C_1 (p_b/w_b) - ||\bs{p}_1||^2_w]. \]
Since we have just shown that $\sqrt{H_2}>\sqrt{H_1}$ must hold, it follows that each $b\in\mcal{B}_1$ satisfies either (i) $C_2 (p_b/w_b) - ||\bs{p}_1||^2_w<C_1 (p_b/w_b) - ||\bs{p}_1||^2_w<0$; or (ii) $C_2 (p_b/w_b) - ||\bs{p}_1||^2_w > C_1 (p_b/w_b) - ||\bs{p}_1||^2_w > 0$.  
% Due to the equality, we cannot have only one of the expressions be negative.  
% Consider scenario (i): $C_2(p_b/w_b) < C_1(p_b/w_b) < ||\bs{p}_1||^2_w$.  Recall that $C_1 = R_A+P_1$, and $P_1=\sum_{b\in\mcal{B}_1} p_b$.  Since the inequality $C_1(p_b/w_b) < ||\bs{p}_1||^2_w$ cannot hold for $b\in\arg\,\max_{b\in\mcal{B}_1} \{p_b/w_b\}$, scenario (ii) must hold for $b\in\arg\,\max_{b\in\mcal{B}_1} \{p_b/w_b\}$, which implies $C_2(p_b/w_b) > C_1(p_b/w_b) > ||\bs{p}_1||^2_w$.  It directly follows that $C_2>C_1$ must hold, and, thus, all $b\in\mcal{B}_1$ must satisfy scenario (ii).  Recall that $C_2=qR_B-P_2$, and, thus, the inequality $C_2>C_1$ is equivalent to $q R_B > R_A+P$.
Observe that $C_1 (p_b/w_b) > ||\bs{p}_1||^2_1$ must hold for $b' \in \arg\,\max_{b\in\mcal{B}_1} p_b/w_b$, and thus $b'$ must satisfy scenario (ii) and $C_2>C_1$ (or, equivalently, $qR_B-P>R_A$). This last inequality then implies that scenario (ii) must be satisfied for all $b\in\mcal{B}_1$.

We have shown that, in order for $\partial \pi_A/(\partial p_b)-\partial \pi_A/(\partial p_c)=0$ to hold for all $b\in\mcal{B}_1$ and $c\in\mcal{B}_2$, a critical point $\bs{p}$ must satisfy
\[ \frac{p_b}{w_b} = \bar p := \frac{\sqrt{H_2} - \sqrt{H_1}}{C_1\sqrt{H_2}-C_2\sqrt{H_1}}||\bs{p}_1||^2_w, \]

\noindent for each $b\in\mcal{B}_1$.  Expanding this expression, and solving for $\bar p$ explicitly, we obtain the following two possible (real) solutions for $\bar p$:
\[ \bar p = 0 \text{ or } \bar p = \frac{2(q R_B - P)(q R_B-R_A-P)}{W R_A}, \]
where we use $P_1=W_1 \bar p$, $P_2=P-P_1$, and $||\bs{p}_1||^2_w=W_1 (\bar p)^2$.  As $\bar p=0$ is inadmissible, we consider the latter expression for $\bar p$.  After inserting this expression for $\bar p$ into the right-hand side of \eqref{eq:case3_mostrestrictive}, where $\max_b\{p_b/w_b\}=\bar p$, we obtain
\begin{align*}
    & \frac{R_A}{2}\left[1 + \sqrt{1+\frac{2W}{R_A}\bar p} \right] \\
    =\> & \frac{R_A}{2} + qR_B-P - \frac{R_A}{2} = qR_B -P,
\end{align*}
which follows since we showed above that $qR_B-P>R_A$ must hold.  Thus, the only critical point sits at the boundary of the region where all battlefields are in Case 2, since decreasing $\bar p$ even slightly will satisfy the condition in \eqref{eq:case3_mostrestrictive}.  We can further verify that the payoff at this critical point is equal to the constant payoff in the region where all battlefields are in Case 2, but omit this for conciseness.

We conclude the proof by resolving the scenario where $\bs{p}$ lies on the boundaries of $\Delta_n(P)$.  Observe that the conditions on $q\kappa^*_B$ and $\kappa^*_A$ immediately imply that $p_b/w_b > p_c/w_c$ for any $b\in\mcal{B}_1$ and $c\in\mcal{B}_2$.  Thus, on the boundaries of $\Delta_n(P)$, it must either be that all battlefields with $p_b=0$ (and possibly more) are in Case 2, or that all battlefields in $\mcal{B}$ are in Case 1 (which is covered by \Cref{lem:local_maximizer}).  In the scenario where all battlefields with $p_b=0$ are in Case 2, note that the necessary condition $(\partial/(\partial p_i)-\partial/(\partial p_j)) \pi_A \geq 0$ for $i\in\arg\,\min_{b\in\mcal{B}_1} \{p_b/w_b\}$ and $j\in\arg\,\max_{b\in\mcal{B}_1} \{p_b/w_b\}$ only holds with equality if $p_b/w_b = P_1/W_1$ for all $b\in\mcal{B}_1$.  If $P_1/W_1 < \bar p$, then the inequality in \eqref{eq:case3_mostrestrictive} is satisfied implying that all battlefields are in Case 2, and Lemma \ref{lem:local_maximizer} shows that $\bs{p}^*=\bs{w} (P/W)$ must correspond with the same payoff to player $A$.  Otherwise, if $P_1/W_1 = \bar p$, then we showed above that the global maximum sits at the boundary where all battlefields are in Case 2 and $\bs{p}^*=\bs{w} (P/W)$ achieves the same payoff.  Finally, if $P_1/W_1 > \bar p$, then, from \eqref{eq:mixed_partials}, we know that $\partial \pi_A/(\partial p_b) - \partial \pi_A/(\partial p_c) < 0$ must hold for all $b\in\mcal{B}_1$ and $c\in\mcal{B}_2$, since the choice $p_b/w_b = \bar p$ satisfies $\partial \pi_A/(\partial p_b) - \partial \pi_A/(\partial p_c) = 0$, and $\partial \pi_A/(\partial p_b)$ is decreasing with respect to $p_b/w_b$ while $\partial \pi_A/(\partial p_c)$ is constant.  This violates a necessary condition for a critical point, and implies that $A$'s payoff is increasing in the direction of decreasing $p_b$ and increasing $p_c$, as expected.  % In the second scenario, when all battlefields are in Case 1, Lemma  shows that $\bs{p}^*=\bs{w} (P/W)$ must correspond with a higher payoff to player $A$.
\qed

%%

% \subsection{Proof of Lemma \ref{lem:optimal_glf_payoffs}}\label{sec:thm_proof}

% In words, this is player $A$'s Stage 2 equilibrium payoff for any pre-allocated and real-time resource pair when the pre-allocated resources are deployed optimally, i.e., $\bs{p}=(P/W)\cdot \bs{w}$.

% \begin{lemma} \label{lem:optimal_glf_payoffs}
%     %  If $qR_B - P \leq R_A + \frac{P R_A}{R_A + \sqrt{R_A(R_A + 2P)}}$, then the equilibrium payoff $\pi_A(P,R_A,R_B)$ is given by
%     If $qR_B - P < 0$, or $qR_B - P \geq 0$ and $R_A \geq \frac{2(qR_B-P)^2}{P+2(qR_B-P)}$, then the optimal payoff $\pi^*_A(R_A,R_B,\bs{w},P)$ is
%     \begin{equation} \label{eq:allcase1_payoff}
%         W\cdot\left(1 - \frac{qR_B}{2R_A}\left(\frac{R_A + \sqrt{R_A(R_A+2P)}}{P+R_A+\sqrt{R_A(R_A+2P)}} \right)^2 \right).
%     \end{equation}
%     Otherwise, it is
%     \begin{equation} \label{eq:allcase2_payoff}
%         W\cdot\frac{R_A}{2(qR_B-P)}.
%     \end{equation}
% \end{lemma}

\smallskip\noindent-- {\bf Part 2:} 
In the proof of \Cref{lem:local_maximizer}, we provide the closed-form solutions to the system of equations \eqref{eq:SOE} for the symmetric case $\mcal{B}_1=\mcal{B}$ (resp. $\mcal{B}_2=\mcal{B}$) in \eqref{eq:case1solution_closedform} (resp. \eqref{eq:case2solution_closedform}).  In the following analysis, we derive conditions on the underlying parameters for which these closed-form solutions of \eqref{eq:SOE} exist and satisfy the corresponding constraints on $\kappa^*_A, q\kappa^*_B>0$, and find that these two cases encompass all possible game instances $\text{GL-P}(P,R_A,R_B,\bs{w})$.

\smallskip\noindent\underline{{\bf Case 1 ($\mcal{B}_1 = \mcal{B}$):}} Substituting $\bs{p}=(P/W)\cdot\bs{w}$ into \eqref{eq:case1solution_closedform} and simplifying, we obtain
\begin{equation}\label{eq:kappa_case2}
    \begin{aligned}
        q\kappa_B^* &= \frac{1}{W}\left[P+R_A + \sqrt{R_A(R_A + 2P)} \right] \\ 
        \kappa_A^* &= \frac{(P+R_A)q\kappa_B^* - P^2/W}{qR_B}.
    \end{aligned}
\end{equation}
Next, we verify that this solution satisfies the conditions $0 < q\kappa_B^* - P/W \leq \kappa_A^*$.

\noindent{$\bullet$ $q\kappa_B^* - P/W > 0$:} This holds by inspection.

\noindent{$\bullet$ $q\kappa_B^* - P/W \leq \kappa_A$:} We can write this condition as
\begin{equation}\label{eq:case2_cond3}
        qR_B - P \leq R_A + \frac{PR_A}{R_A + \sqrt{R_A(R_A + 2P)}}
\end{equation}
We note that whenever $qR_B - P < 0$, this condition is always satisfied. When $qR_B - P \geq 0$, this condition does not automatically hold, and an equivalent expression of \eqref{eq:case2_cond3} is given by
\begin{equation}
    R_A \geq \frac{2(qR_B-P)^2}{P + 2(qR_B-P)}.
\end{equation}
Observe that $R_A = \frac{2(qR_B-P)^2}{P + (qR_B-P)}$ satisfies \eqref{eq:case2_cond3} with equality, and is in fact the only real solution (one can reduce it to a cubic polynomial in $R_A$). 

When these conditions hold, the equilibrium payoff $\pi_A(P,R_A,R_B)$ (computed from Lemma \ref{lem:SOE}) is  given by the expression \eqref{eq:allcase1_payoff}.

% $$W\cdot\left(1 - \frac{qR_B}{2R_A}\left(\frac{R_A + \sqrt{R_A(R_A+2P)}}{P+R_A+\sqrt{R_A(R_A+2P)}} \right)^2 \right).$$

\smallskip\noindent\underline{{\bf Case 2 ($\mcal{B}_2 = \mcal{B}$):}} Substituting $\bs{p}=(P/W)\cdot\bs{w}$ into \eqref{eq:case2solution_closedform} and simplifying, we obtain
\begin{equation}
    \kappa_A^* = \frac{2(qR_B - P)}{W} \text{ and } q\kappa_B^* = \frac{2(qR_B - P)^2}{WR_A}.
\end{equation}
The solution satisfies the condition $0 < \kappa_A^* < q\kappa_B^* - P/W$ if and only if $qR_B - P > 0$ and  $R_A > \frac{2(qR_B-P)^2}{P + (qR_B-P)}$. When this holds, the equilibrium payoff (from Lemma \ref{lem:SOE}) is $\pi_A(P,R_A,R_B) = W\cdot \frac{R_A}{2(qR_B-P)}$. \qed

\end{document}